\documentclass[a4paper, 12pt]{scrartcl}
\usepackage[english]{babel}
\usepackage[ansinew]{inputenc}
\usepackage{amsmath}
\usepackage[T1]{fontenc}
\usepackage{enumerate}
\usepackage[all]{xy}
\usepackage{amsmath,amssymb,amsthm,enumerate}
\usepackage{graphicx}
\usepackage{epsfig,psfrag}
\usepackage{bbm}

\newtheorem{lemma}{Lemma}
\newtheorem{remark}{Remark}
\newtheorem{theorem}{Theorem}
\newcommand{\EE}{\mathbb{E}}
\newcommand{\NN}{\mathbb{N}}
\newcommand{\PP}{\mathbb{P}}

\newcommand{\bM}{\boldsymbol{M}}

\newcommand{\cA}{\mathcal{A}}

\title{The common ancestor process revisited}
\author{Sandra Kluth,
Thiemo Hustedt,
Ellen Baake \\
\small{Technische Fakult\"at,
Universit\"at Bielefeld,
Box 100131,
33501 Bielefeld, Germany } \\
\small{E-mail: \{skluth, thustedt, ebaake\}\emph{@}techfak.uni-bielefeld.de}
}
\date{}

\begin{document}
 
\maketitle

\begin{abstract}
\textbf{Abstract.} We consider the Moran model in continuous time with two types, mutation,
and selection. We concentrate on 
the ancestral line
and its stationary type distribution. Building on work by Fearnhead (J. Appl. Prob.
39 (2002), 38-54) and Taylor (Electron. J. Probab. 12 (2007), 808-847),
we characterise this distribution via the fixation probability of the offspring
of all individuals of favourable type (regardless of the offspring's types).
We concentrate on a finite population and stay with the resulting discrete
setting all the way through. 
This way, we extend previous results and gain new insight into the underlying particle picture. 


\end{abstract}

\section{Introduction}
Understanding the interplay of random reproduction, mutation, and selection
is a
major topic of population genetics research. In line with the historical
perspective of
evolutionary research, modern approaches aim at tracing back the ancestry
of a sample
of individuals taken from a present population. Generically, in populations
that evolve 
at constant size over a long time span without recombination,
the ancestral lines
will eventually coalesce backwards in time into a single line of descent.
This ancestral line is of special
interest. In particular, its type composition
may differ substantially from the distribution at present time.
This mirrors the fact that the ancestral
line consists of those
individuals that are successful in the long run; thus, its type distribution
is
expected to be biased towards the favourable types.

This article is devoted to the ancestral line
in a classical model of population
genetics, namely, the Moran model in continuous time with two types,
mutation, and selection (i.e., one type is `fitter' than the other).
We are particularly interested in the
stationary distribution of the types along the ancestral line, 
to be called the \emph{ancestral type distribution}.
We build on previous work by Fearnhead
\cite{Fearnhead} and Taylor \cite{Taylor}. Fearnhead's approach is
based on the \emph{ancestral selection graph}, or ASG for short
\cite{Krone,Neuhauser}. The ASG is an
extension of Kingman's coalescent \cite{Kingman1,Kingman2},
which is the central tool
to describe the genealogy of a finite sample in the \emph{absence}
of selection. The ASG copes with selection by including so-called
\emph{virtual branches} in addition to the \emph{real branches}
that define the true genealogy.
Fearnhead calculates the ancestral type distribution in terms of
the coefficients of a series expansion that is related to the number
of (`unfit') virtual branches.
 
Taylor uses diffusion theory and a backward-forward construction that
relies on a description of the \emph{full} population. He
characterises the ancestral type
distribution  in terms of the fixation probability of
the offspring of all `fit' individuals (regardless of the offspring's types).
This fixation probability is calculated via a boundary value problem.

Both approaches rely strongly on analytical tools; in particular,
they employ the diffusion limit (which assumes infinite population
size, weak selection and mutation)
from the very beginning. The results only have partial
interpretations in terms of the graphical representation of the model
(i.e., the representation that makes individual lineages and their
interactions explicit). The aim of this article is to complement these
approaches by starting from the graphical representation for a population
of finite size and staying with the resulting discrete setting all the
way through, performing the diffusion limit only at the very end. This will
result in an extension of the results to arbitrary selection strength, as 
well as a number of new insights, such as an intuitive explanation of
Taylor's boundary value problem in terms of the particle picture, and
an alternative derivation of the ancestral type distribution.

The paper is organised as follows. We start with a short outline of the
Moran model (Section 2). In Section 3, we introduce the common ancestor type
process and briefly recapitulate Taylor's and Fearnhead's approaches.
We concentrate on a Moran model of finite size and trace the descendants
of the initially `fit' individuals forward in time. Decomposition
according to what can happen after the first step gives a difference
equation, which turns into Taylor's diffusion equation in the limit.
We solve this difference equation and obtain the fixation probability
in the finite-size model in closed form. In Section 5, we derive the
coefficients of the ancestral type distribution within the discrete
setting. Section 6 summarises and discusses the results.

\section{The Moran model with mutation and selection}
\label{The Moran model with mutation and selection}
We consider a haploid population of fixed size $N \in \mathbb{N}$ in 
which each individual is characterised by a type $i \in S = \{0,1\}$. 
If an individual reproduces, its single offspring inherits the parent's 
type and replaces a randomly chosen individual, maybe its own parent. This
way the replaced individual dies and the population size remains constant. 

Individuals of type $1$ reproduce at rate $1$, whereas individuals of type 
$0$ reproduce at rate $1+s^{}_N$, $s^{}_N \geqslant 0$. 
Accordingly, type-$0$ individuals are termed  `fit', type-$1$ individuals are
`unfit'. 
In line with a central idea of the ASG, we will
decompose reproduction events into neutral and selective ones. Neutral ones
occur at rate $1$ and happen to all individuals, whereas  selective events 
occur at rate $s_N^{}$ and are
reserved for type-$0$ individuals. 

Mutation occurs independently of reproduction. An individual of type
$i$ mutates to type $j$ at rate $u^{}_N \nu^{}_j$, $u^{}_N \geqslant 0$, $0
\leqslant \nu^{}_j \leqslant 1$, $\nu^{}_0 + \nu^{}_1 = 1$. This is to be understood
in the sense that every individual, regardless of its type, mutates at
rate $u^{}_N$ and the new type is $j$ with probability
$\nu^{}_j$. Note that this includes the possibility of `silent'
mutations, i.e., mutations from
type $i$ to type $i$. 

The Moran model has a well-known graphical representation as an
interacting particle system (cf. Fig. \ref{moran model}).
The $N$ vertical lines represent the $N$ individuals, with time running from 
top to bottom in the figure. Reproduction events are represented by arrows with the 
reproducing individual at the base and the offspring at the tip. 
Mutation events are marked by bullets.

\begin{figure}[ht]
\begin{center}
\psfrag{t}{\Large$t$}
\psfrag{0}{\Large$0$}
\psfrag{1}{\Large$1$}
\includegraphics[angle=0, width=7cm, height=5cm]{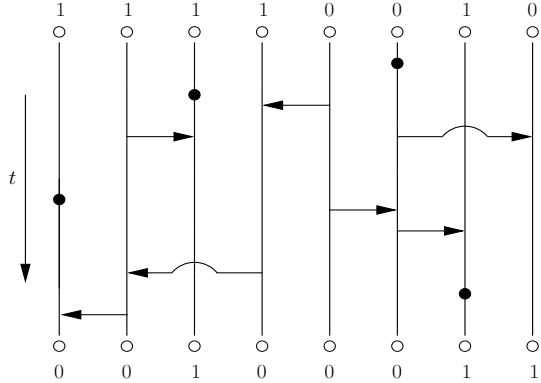}
\caption{The Moran model. The types ($0=$ fit, $1=$ unfit) are
indicated for the initial population (top) and the final one (bottom).}
\label{moran model}
\end{center}
\end{figure} 

We are now interested in the process $\left( Z_{t}^N \right)_{t \geqslant 0}$, where 
$Z_{t}^N$ is the number of individuals of type $0$ at time $t$. When the number
of type-$0$ individuals is $k$, it increases by one at rate 
$\lambda_k^{N}$ and decreases by one at rate $\mu_k^{N}$, where
\begin{equation}\label{eq:lambda_mu}
\lambda_k^{N}=\frac{k(N-k)}{N}(1+s_N^{})+(N-k)u_N^{}\nu_0^{}\quad \text{and}
\quad \mu_k^{N}=\frac{k(N-k)}{N}+ku_N^{}\nu_1^{}. 
\end{equation}
Thus, $\left( Z_{t}^N \right)_{t \geqslant 0}$ is a birth-death process with birth
rates $\lambda_k^{N}$ and death rates $\mu_k^{N}$. For $u_N^{} >0$ and $0 < \nu_0^{}, \nu_1^{} <1$ its stationary distribution
is $\left( \pi^N_Z \left(k\right) \right)_{0 \leqslant k \leqslant N}$ with
\begin{equation}
 \pi^N_Z \left(k\right) = C^{}_N  \prod_{i=1}^k \frac{\lambda_{i-1}^{N}}{\mu_i^{N}}  , \quad 0 \leqslant k \leqslant N,
 \label{stationaere Verteilung Moran}
\end{equation}
where $C^{}_N$ is a normalising constant (cf. \cite[p. 19]{Durrett1}). (As usual, an empty 
product is understood as $1$.)

To arrive at a diffusion, we consider the usual rescaling 
\begin{equation*}
\left( X_{t}^N  \right)_{t \geqslant 0} := \frac{1}{N} \left( Z_{Nt}^N  \right)_{t \geqslant 0},
\end{equation*}
and assume that $\lim_{N \to \infty}Nu_N^{} = \theta$, $0 < \theta < \infty$,
and $\lim_{N \to \infty}Ns_N^{} = \sigma$, $0 \leqslant \sigma < \infty$. As $N \to \infty$,
we obtain the well-known diffusion limit 
\begin{equation*}
\left(X_t^{} \right)_{t \geqslant 0} := \lim_{N \to \infty} \left( X_{t}^N  \right)_{t \geqslant 0}.
\end{equation*}
Given $x \in [0,1]$, a sequence $\left( k_N^{} \right)_{N \in \mathbb{N}}$ with
$k_N^{} \in \lbrace 0, \dots ,N \rbrace $ and $\lim_{N \to \infty} \frac{k_N^{}}{N} = x$,
$\left( X_t^{} \right)_{t \geqslant 0}$ is characterised by the drift coefficient
\begin{equation}
 a(x) = \lim_{N \to \infty} (\lambda_{k_N^{}}^{N} - \mu_{k_N^{}}^{N}) 
 = (1-x)\theta \nu^{}_0 - x \theta \nu_1^{} + (1-x)x \sigma
\label{drift coefficient}
\end{equation}
and the diffusion coefficient
\begin{equation}
 b(x)= \lim_{N \to \infty} \frac{1}{N} \left(\lambda_{k_N^{}}^{N} + \mu_{k_N^{}}^{N}  \right)  
 = 2 x (1-x).
\label{diffusion coefficient}
\end{equation}
Hence, the infinitesimal generator $A$ of the diffusion is defined by
\begin{equation*}
Af(x)= (1-x) x \frac{\partial^2{}}{\partial{x^2}} f(x)  + \left[ (1-x)\theta \nu^{}_0 - x \theta \nu_1^{} + (1-x)x \sigma \right] \frac{\partial{}}{\partial{x}} f(x) , \ f \in \mathcal{C}^2_{}([0,1]).
\end{equation*}
The stationary density $\pi^{}_X$ -- known as Wright's formula -- is given by
\begin{equation}\label{wright}
 \pi^{}_X(x)= C (1-x)^{\theta \nu^{}_1 - 1} x^{\theta \nu^{}_0 - 1} \exp (\sigma x),
\end{equation}
where $C$ is a normalising constant. 
See \cite[Ch. 7, 8]{Durrett} or \cite[Ch. 4, 5]{Ewens} for reviews of diffusion processes
in population genetics and \cite[Ch. 15]{Karlin} for a general survey of diffusion theory. 

In contrast to our approach starting from the Moran model, 
\cite{Fearnhead} and \cite{Taylor} choose the
diffusion limit of the \emph{Wright-Fisher model} as the basis for the common 
ancestor process. This is, however, of minor importance, 
since both diffusion limits differ 
only by a rescaling of time by a  factor of $2$ (cf. \cite[Ch. 7]{Durrett}, \cite[Ch. 5]{Ewens} or 
\cite[Ch. 15]{Karlin}).

\section{The common ancestor type process}
\label{The common ancestor type process}
Assume that the population is stationary and evolves according to the
diffusion process $\left(X_t^{} \right)_{t \geqslant 0} $. Then, at any time
$t$, there almost surely exists a unique individual that is, at some time $s > t$,
ancestral to the whole population; cf. Fig.~\ref{aline}. (One way to see
this is via \cite[Thm. 3.2, Corollary
3.4]{Krone}, which shows that the expected time to the ultimate ancestor
in the ASG remains bounded if the sample size tends to infinity.) We say that the descendants of this individual become
fixed and call it the \textit{common ancestor at time $t$}. The lineage of these
distinguished individuals over time defines the so-called \textit{ancestral
  line}. Denoting the type of the common ancestor at time $t$ by $ I_t
$, $I_t \in S$, we term $\left( I_t \right)_{t \geqslant 0}$ the
\textit{common ancestor type process} or \textit{CAT process} for short. Of
particular importance is its stationary type distribution
$\alpha=\left( \alpha^{}_i \right)_{i \in S}$, 
to which we will
refer as the \emph{ancestral type distribution}. Unfortunately,
the CAT process is not Markovian.  But two approaches
are available that augment $\left( I_t \right)_{t \geqslant 0}$ by a
second component to obtain a Markov process. They go back to Fearnhead
\cite{Fearnhead} and Taylor \cite{Taylor}; we will recapitulate them
below.
\begin{figure}[ht]
\begin{center}
\psfrag{t}{\Large$t$}
\psfrag{s}{\Large$s$}
\psfrag{t-h}{\Large$t - \tau_0^{}$}
\psfrag{CA}{\Large CA}
\includegraphics[angle=0, width=7cm, height=5cm]{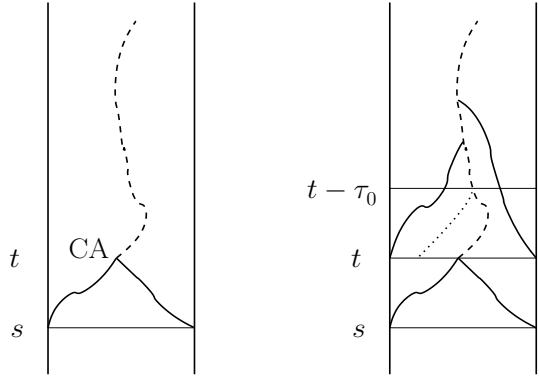}
\caption{Left: The common ancestor at time $t$ (CA) is the individual whose progeny 
will eventually fix in the population (at time $s > t$). Right: If we pick an arbitrary 
individual at time $t$, there exists a
minimal time $\tau_0^{}$ so that the individual's line of ancestors (dotted)
corresponds to the ancestral line (dashed) up to time $t - \tau_0^{}$.}
\label{aline}
\end{center}
\end{figure}

\subsection{Taylor's approach}
For ease of exposition, we start with Taylor's approach \cite{Taylor}.
It relies on a description of the \emph{full} population forward in time (in the  diffusion limit
of the
Moran model as $N \to \infty$) and builds on the so-called \textit{structured coalescent} \cite{Barton}. 
The process is  $\left( I_t, X_t \right)_{t \geqslant 0}$, 
with states $(i,x)$, $i \in S$ and $x \in [0,1]$. In \cite{Taylor} this
process is termed \textit{common ancestor process (CAP)}. 

Define $h(x)$ as the probability that the common ancestor at a
given time is of type $0$, provided that the frequency of type-$0$ individuals
at this time is $x$. Obviously, $h(0)=0$, $h( 1) = 1$. Since the 
process is time-homogeneous, $h$ is independent of time.
Denote the (stationary) distribution of $\left( I_t, X_t \right)_{t \geqslant 0}$
by $\pi^{}_T$. Its marginal distributions are $\alpha$
(with respect to the first variable) and
$\pi^{}_X$ (with respect to the second variable).
 $\pi^{}_T$ may then be written as
the product of the marginal density $\pi^{}_X(x)$ and the 
conditional probability $h(x)$ (cf. \cite{Taylor}):
\begin{align*}
\pi^{}_T\left(0, x \right) dx &= h( x ) \pi^{}_X( x ) dx, \displaybreak[0] \\
\pi^{}_T\left(1, x \right) dx &= \left(1-h( x )\right) \pi^{}_X( x ) dx.
\end{align*}
Since $\pi^{}_X$ is well known \eqref{wright}, it remains to specify $h$. 
Taylor uses a backward-forward construction within diffusion theory to derive a 
boundary value problem for $h$,
namely:
\begin{equation}
\begin{split}
 &\frac{1}{2} b(x)  h''  (x) + a(x) h'(x)  
 - \Big ( \theta \nu_1^{} \frac{x}{1 - x} + \theta \nu_0^{} \frac{1-x}{x} 
\Big ) h ( x ) + \theta \nu_1^{} \frac{x}{1-x} = 0, \\
&h(0)=0, h(1)=1.
\label{h DGL}
\end{split}
\end{equation}
Taylor  shows that \eqref{h DGL} has a unique solution. The stationary 
distribution of $\left( I_t, X_t \right)_{t \geqslant 0}$ is thus determined in a 
unique way as well. The function $h$ is smooth in $(0,1)$ and its derivative $h'$ 
can be continuously extended to $[0,1]$ (cf. \cite[Lemma 2.3, Prop. 2.4]{Taylor}). 

In the neutral case (i.e.,  without selection, $\sigma = 0$), all individuals 
reproduce at the same rate, independently of their types.
 For reasons of symmetry, the  common ancestor thus is a uniform random draw
from the population; consequently, $h ( x ) = x$. 
In the presence of selection, Taylor determines the solution of the 
boundary value problem via a series expansion in $\sigma$
(cf. \cite[Sec. 4]{Taylor} and see below), which yields
\begin{align}
 &h(x)= x + \sigma x^{-\theta \nu_0^{}}_{} \left( 1-x \right)^{- \theta \nu_1^{}}_{} \exp ( - \sigma x ) \int_0^x \left( \tilde{x} - p \right) p^{\theta \nu_0^{} }_{} \left( 1-p \right)^{\theta \nu_1^{}}_{} \exp( \sigma p ) dp \label{solution h}\\
&\text{with  }\tilde{x}= \frac{\int_0^1 p_{}^{\theta \nu_0^{} + 1} \left( 1-p \right)^{\theta \nu_1^{}}_{} \exp ( \sigma p ) dp}{\int_0^1 p_{}^{\theta \nu_0^{}} \left( 1-p \right)^{\theta \nu_1^{}}_{} \exp ( \sigma p ) dp} = 
\frac{\EE_{\pi_X^{}}(X^2 (1-X))}{\EE_{\pi_X^{}}(X (1-X))}. \label{x schlange}
\end{align}
The stationary type distribution of the ancestral line now
follows  via marginalisation:
\begin{equation}
 \alpha_0^{} = \int_0^1 h ( x ) \pi_X^{}( x ) dx \text{ \ and \ } \alpha_1^{} = \int_0^1 \left( 1-h (x ) \right) \pi_X^{} ( x ) dx .
\label{alpha}
\end{equation} 
Following \cite{Taylor}, we define $\psi (x ):= h(x) - x$ and write
\begin{equation}
 h (x ) = x + \psi ( x ).
\label{def psi}
\end{equation}
Since $h(x )$ is the conditional probability that the common ancestor is fit, $\psi (x)$ 
is the part of this probability that is due to selective reproduction.  \\
Substituting (\ref{def psi}) into (\ref{h DGL}) leads to a boundary value problem
for $\psi$:
\begin{equation}
\begin{split}
& \frac{1}{2} b(x) \psi'' (x) + a(x) \psi' \left( x \right) 
 - \left( \theta \nu^{}_1 \frac{x}{1 - x} + \theta \nu^{}_0 \frac{1-x}{x} \right) \psi(x)+\sigma x \left( 1 - x \right)=0, \\
 &\psi(0)= \psi(1)=0. 
\label{psi DGL}
\end{split}
\end{equation}
Here, the smooth inhomogeneous term is more favourable as compared to the divergent 
inhomogeneous term in (\ref{h DGL}). Note that Taylor actually derives the boundary value problems \eqref{h DGL}
and \eqref{psi DGL} for the more general case of frequency-dependent selection,
but restricts himself to frequency-independence to derive solution \eqref{solution h}.

\subsection{Fearnhead's approach}
\label{Fearnhead's approach}
We can only give a brief introduction to Fearnhead's approach
\cite{Fearnhead} here. On the basis of the ASG, the
ancestry of a randomly chosen individual from the present (stationary)
population
is traced backwards in time. More precisely, one considers
the process $(J_{\tau})_{\tau \geqslant 0}$ with values in $S$, where
$J_{\tau}$ is the type of the individual's ancestor
at time $\tau$ before the present (that is, at forward time $t - \tau$).
Obviously, there is a minimal time $\tau_0^{}$ so that, for all
$\tau \geqslant \tau_0^{}$, $J_{\tau}=I_{t-\tau}$ (see also Fig.~\ref{aline}),
provided the underlying  process $(X_t)_{t \geqslant 0}$ is extended to $(-\infty, \infty)$.

To make the process Markov, the true ancestor (known as the
\textit{real} branch) is augmented by a collection of \textit{virtual}
branches (see \cite{Baake,Krone,Neuhauser,Stephens} for the
details). Following \cite[Thm. 1]{Fearnhead}, certain virtual branches
may be removed (without compromising the Markov property) and the
remaining set of virtual branches contains only unfit ones. We will
refer to the resulting construction as the \textit{pruned ASG}. It is
described by the process $(J_{\tau}, V_{\tau})_{ \tau \geqslant 0}$, where $V_{\tau}$ (with values in $\NN_0$) is the
number of virtual branches (of type $1$). $(J_{\tau}, V_{\tau})_{ \tau \geqslant 0}$ is termed \textit{common ancestor process} in \cite{Fearnhead}
(but keep in mind that it is $(I_t, X_t)$ that is called CAP in \cite{Taylor}).
Reversing the direction of
time in the pruned ASG yields an alternative augmentation of the CAT
process (for $\tau  \geqslant \tau_0^{}$). 

Fearnhead provides a representation of the 
stationary distribution of the pruned process, which we will  denote by
$\pi_F^{}$. This stationary distribution is expressed in terms of
constants $\rho_1^{(k)}, \dots , \rho_{k+1}^{(k)}$ defined
by the following \textit{backward} recursion: 
\begin{equation}\label{lambdaback}
 \rho_{k+1}^{(k)}=0 \text{ \ and \ } \rho_{j-1}^{(k)}
 = \frac{\sigma}{j + \sigma + \theta - (j + \theta\nu_1) \rho_{j}^{(k)}}, \  k \in \mathbb{N}, 
2 \leqslant j \leqslant k+1.
\end{equation}
The limit $\rho^{}_j := \lim_{k \rightarrow \infty} \rho_j^{(k)}$ exists 
(cf. \cite[Lemma 1]{Fearnhead}) and  the stationary distribution of 
the pruned ASG is given by (cf. \cite[Thm. 3]{Fearnhead})
\begin{align*}
&\pi^{}_F(i, n)=
\begin{cases}
a_n \EE_{\pi^{}_X} (X(1-X)^n),
   & \text{if} \ \ i = 0 ,\\
(a_n-a_{n+1})\EE_{\pi^{}_X} ((1-X)^{n+1}) , 
& \text{if} \ \ i = 1,
\end{cases} \\
& \text{with \  } a_n:=  \prod_{j=1}^n \rho_j
\end{align*}
for all $n \in \NN_0$.
Fearnhead proves this result by straightforward verification
of the stationarity condition; the calculation is somewhat cumbersome and
does not yield insight into the connection with the graphical
representation of the pruned ASG. Marginalising over the number of virtual 
branches results in the stationary type distribution of the ancestral line, namely,
\begin{equation}
 \alpha^{}_i= \sum_{n \geqslant 0} \pi^{}_F (i, n).
\label{alpha Fearnhead}
\end{equation}
Furthermore, this reasoning points to an alternative representation of $h$ respectively $\psi$ (cf. \cite{Taylor}):
\begin{equation}
 h(x)= x + x \sum_{n \geqslant 1} a_n (1-x)^n \text{ \ respectively \ } \psi(x)= x \sum_{n \geqslant 1} a_n (1-x)^n .  
\label{h solution summe} \\
\end{equation}
The $a_n$, to which we will refer as \emph{Fearnhead's coefficients},
can be shown \cite{Taylor}
to follow the second-order \emph{forward} recursion
\begin{equation}
\begin{split}
 &  \left( 2 + \theta \nu^{}_1 \right) a_2 - \left( 2 + \sigma + \theta  \right)a_1 + \sigma = 0  ,\\ 
& \left( n + \theta \nu^{}_1  \right)a_n - \left( n + \sigma + \theta \right)a_{n-1} + \sigma a_{n-2}=0, \quad n \geqslant 3 .
\label{a_n}
\end{split}
\end{equation}
Indeed, (\ref{h solution summe}) solves the boundary problem
\eqref{h DGL} and,
therefore, equals \eqref{solution h} (cf. \cite[Lemma 4.1]{Taylor}).

The forward recursion \eqref{a_n} is greatly
preferable to the backward recursion \eqref{lambdaback}, which can
only be solved approximately with initial value $\rho_n^{} \approx 0$ 
for some large $n$. What is still missing is the initial value, $a_1$.
To calculate it, Taylor defines (cf. \cite[Sec. 4.1]{Taylor})
\begin{equation}
 v(x):= \frac{h(x)-x}{x} = \frac{\psi(x)}{x}= \sum_{n \geqslant 1}  a_n (1-x)^n 
 \label{v}
\end{equation}
and uses\footnote{Note the missing factor of $1/n$ in his equation (28).}
\begin{equation}
a_n = \frac{(-1)^n}{n!} v^{(n)}_{}(1).
\label{a_n taylor}
\end{equation}
This way a straightforward (but lengthy)
calculation (that includes a differentiation of expression \eqref{solution h})
yields
\begin{equation}
 a_1 = -v'(1)=  - \psi'(1)= \frac{\sigma}{1 + \theta \nu_1^{}} (1-\tilde{x}).
\label{a_1}
\end{equation}

\section{Discrete approach}
\label{Alternative derivation}
Our focus is on the stationary type distribution
$\left(\alpha_i^{}\right)_{i \in S}$ of the CAT process.
We have seen so far that it corresponds to the marginal distribution
of both $\pi_T^{}$ and $\pi_F^{}$, with respect to the first variable.
Our aim now is to establish a closer connection between
the properties of the ancestral type distribution and the graphical
representation of the Moran model. 
In a first step we re-derive the differential equations for $h$ and $\psi$
in a direct way, on the basis of the particle picture for a finite population.
This derivation will be elementary and, at the same time, it will 
provide a direct interpretation of the resulting differential equations.

\subsection{Difference and differential equations for $h$ and $\psi$}
\label{Derivation of differential equations}
\textbf{Equations for $h$.}  Since it is essential to make
the connection with the graphical representation explicit, we start
from a population of finite size $N$, rather than  from the
diffusion limit. Namely, we look at a new Markov process
$(\bM_t^{},Z^N_t)_{t \geqslant 0}$ with the natural filtration $(\mathcal{F}_t^N)_{t \geqslant 0}$, where 
$\mathcal{F}_t^N := \sigma ((\bM_s^{},Z^N_s) \mid s \leqslant t)$.
$Z^N_t$ is the number of fit
individuals as before and $\bM_t^{}=(M_0^{},M_1^{})_t$ 
holds the number of descendants of types $0$ and $1$ at time $t$ of 
an unordered sample with composition $\bM_0^{}=(M_0,M_1)_0$ collected at 
time $0$. More precisely, we
start with a $\mathcal{F}_0^N$-measurable state $(\bM_0,Z^N_0)= (\boldsymbol{m},k)$ (this  means that
$\bM_0$ must  be independent of the future evolution; but note that it
need not be a random sample) and observe the
population evolve in forward time. At time $t$, count the
type-$0$ descendants and the type-$1$ descendants of our initial sample
$\boldsymbol{M}_0$ and summarise the results in the unordered
sample $\boldsymbol{M}^{}_t$. Together with $Z^N_t$, this gives the
current state $(\bM_t,Z^N_t)$
(cf. Fig. \ref{forward}).
\begin{figure}[ht]
\begin{center}
\psfrag{t}{\large $t$}
\psfrag{0}{\large $0$}
\psfrag{1}{\large $1$}
\psfrag{((2,2), 3)}{\Large $\left( (2,2), 3 \right)$}
\psfrag{((1,3), 2)}{\Large $\left( (1,3), 2 \right)$}
 \includegraphics[angle=0, width=8.5cm, height=5.95cm]{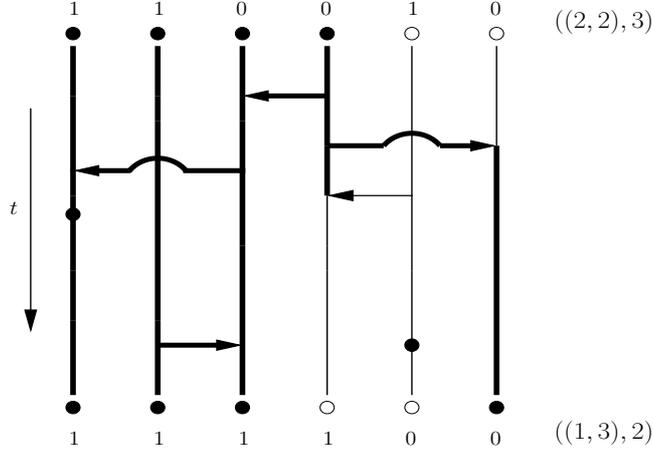}
 \caption{The process $ \left(\boldsymbol{M}^{}_t, Z^N_t \right)_{t
     \geqslant 0}$. The initial sample $\boldsymbol{M}^{}_0=(2,2)$ in a
   population of size $N=6$ (whose number of type-$0$ individuals
   is $Z^N_0=3$) is marked black at the top. Fat lines represent
   their descendants. At the later time (bottom), the descendants
   consist of one type-$0$ individual and three type-$1$ individuals,
   the entire population has two individuals of type $0$. The initial
   and final states of the process are noted at the right.}
 \label{forward}
\end{center}
\end{figure}

As soon as the initial sample is ancestral to all $N$ individuals, it
clearly will be ancestral to all $N$ individuals at all later
times. Therefore,
\begin{equation*}
 \mathcal{A}^{}_N:=\left\{(\boldsymbol{m}, k) : 
 k \in \{0, \dots, N\}, m_0 \leqslant k, \lvert \boldsymbol{m} \rvert= N  \right\},
\end{equation*}
where $\lvert \boldsymbol{m} \rvert=m_0 + m_1$ for a sample $\boldsymbol{m}=(m_0,m_1)$,
is a closed (or invariant) set of the Markov chain. (Given a Markov
chain $( Y ( t) )_{t \geqslant 0}$ in continuous time on a discrete state
space $E$, a non-empty subset $\cA\subseteq E$ is called closed (or
invariant) provided that $\mathbb P (Y(s)=j \mid Y(t)=i)=0$ $\forall s
 > t$, $i \in \cA$, $j \notin \cA$ (cf. \cite[Ch. 3.2]{Norris}).) 

 From now on we restrict ourselves to the initial value
 $(\bM_0,Z^N_0)=\left( (k, 0), k \right)$, i.e. the
 population consists of $k$ fit individuals and the initial sample
 contains them all.  
Our aim is to calculate the probability of absorption in
 $\mathcal{A}^{}_N$, which will also give us the fixation probability of the descendants
 of the type-$0$ individuals. 
 In other words, we are interested in the
 probability that the common ancestor at time 0 belongs to our fit sample
 $\boldsymbol{M}_0^{}$. Let us define $h^N_{}$ as the equivalent of $h$ in
 the case of finite population size $N$, that is, $h^N_{k} $ is the
 probability that one of the $k$ fit individuals is the common
 ancestor given $Z^N_0=k$. Equivalently, $h_k^N$ is the
 absorption probability of $(\bM_t, Z^N_t)$ in $\mathcal{A}^{}_N$,
 conditional on $(\bM_0, Z^N_0)=\left( (k, 0), k
 \right)$. Obviously, $h^N_0 = 0$, $h^N_N=1$. It is important to note
 that, given absorption in $\cA^N_{}$, 
the common ancestor is a random draw from the initial sample. Therefore,
\begin{equation}\label{fitfix}
\PP\big(\text{a specific type-0 individual will fix} \mid Z^N_0=k\big) =
\frac{h_k^N}{k}.
\end{equation}
Likewise,
\begin{equation}\label{unfitfix}
\PP\big(\text{a specific type-1 individual will fix} \mid Z^N_0=k \big) =
\frac{1-h_k^N}{N-k}.
\end{equation}

 We will now calculate the absorption probabilities with the
 help of `first-step analysis' (cf. \cite[Thm. 3.3.1]{Norris}, see also \cite[Thm. 7.5]{Durrett}).
Let us recall the
method for convenience.
\begin{lemma}[`first-step analysis'] Assume that $\left( Y \left( t
\right) \right)_{t \geqslant 0}$ is a Markov chain in continuous time on a
discrete state space $E$, $\cA\subseteq E$ is a closed set and $T_x$,
$x \in E$, is the waiting time to leave the state
$x$. Then for all $y \in E$,
\begin{align*} 
\mathbb{P}\left(  Y \ \textnormal{absorbs
in} \ \cA \mid  Y( 0 )=y \right) &=\sum_{z \in E: z \neq y}
\mathbb{P}\left( Y ( T_y ) = z \mid Y ( 0 \right) = y ) \\
&\hphantom{=\sum_{z \in E: z \neq x} } \times \mathbb{P} \left( Y
\ \textnormal{absorbs in} \ \cA \mid Y ( 0 ) =z
\right).  
\end{align*} \label{first step} 
\end{lemma}
So let us decompose the event `absorption in $\mathcal{A}^{}_N$'
according to the first step away from
the initial state. Below we analyse all possible transitions (which are
illustrated in Fig. \ref{transitions}), state the transition rates and 
calculate absorption probabilities, based upon the new state. We assume 
throughout that $0<k<N$.
\begin{figure}[ht]
\begin{center}
\psfrag{(a)}{(a)}
\psfrag{(b)}{(b)}
\psfrag{(c)}{(c)}
\psfrag{(d)}{(d)}
\psfrag{t}{\small $t$}
 \includegraphics[angle=0, width=7.2cm, height=6.8cm]{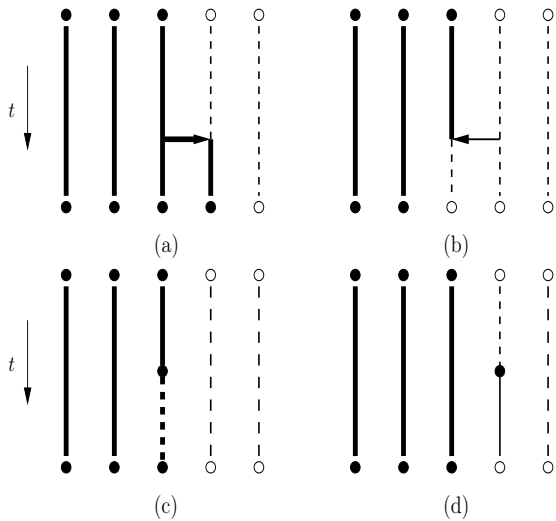}
 \caption{Transitions out of $\left((k,0), k\right)$. Solid
   lines represent type-$0$ individuals, dashed ones type-$1$
   individuals. Descendants of type-$0$ individuals (marked black at
   the top) are represented by bold lines.}
 \label{transitions}
\end{center}
\end{figure}
\begin{description}

  \item[(a)] \textit{$\left((k,0),k\right) \rightarrow
    \left((k + 1, 0), k+1\right)$}:

  One of the $k$ sample individuals of type $0$ reproduces and
  replaces a type-$1$ individual. We distinguish according to the kind of
  the reproduction event.
  \begin{description}
 \item [(a1)] Neutral reproduction rate: $\frac{k ( N- k )}{N}$.
\item[(a2)] Selective reproduction rate: $\frac{k ( N- k )}{N}s_N^{}$.
\end{description}
In both cases, the result is a sample containing all $k + 1$ fit
individuals. Now $(\bM_t,Z^N_t)$ starts afresh in the new state $\left((k
+ 1, 0), k+1\right)$, with absorption probability $h^{N}_{k +
  1}$.

\item[(b)] \textit{$ \left((k,0),k \right) \rightarrow
  \left((k - 1, 0), k-1\right) $ }: \\ A type-$1$ individual
  reproduces and replaces a (sample) individual of type $0$. This
  occurs at rate $\frac{k (N-k)}{N}$ and leads to a sample that
  consists of all $k - 1$ fit individuals. The absorption
  probability, if we start in the new state, is $h^N_{k-1}$.

\item[(c)] \textit{$ \left((k,0),k\right) \rightarrow
  \left((k - 1, 1), k-1\right) $}:\\ This transition
  describes a mutation of a type-$0$ individual to type $1$ and occurs
  at rate $ku^{}_N\nu^{}_1$. The new sample contains all $k - 1$ fit
  individuals, plus a single unfit one. Starting now from
  $\bigl((k - 1, 1), k-1\bigr)$, the absorption probability
  has two contributions: First, by definition, with probability $
  h^N_{k-1} $, one of the $k -1$ fit individuals will be the common
  ancestor. In addition, by \eqref{unfitfix}, the single unfit
  individual has fixation probability $(1-h^N_{k-1})/(N-(k-1))$, so
  the probability to absorb in $\mathcal{A}^{}_N$ when starting from
  the new state is
\begin{align*}
&\mathbb{P} \left(  \text{absorption in} \ \mathcal{A}_N^{}  \mid \left(\boldsymbol{M}^{}_0,Z^N_0  \right) 
= \left((k - 1, 1), k-1\right)\right) \\
&=h^N_{k-1} + \frac{1-h^N_{k-1}}{N-(k-1)} .
\end{align*}

\item[(d)] \textit{$   \left((k,0),k\right) \rightarrow \left((k , 0), k+1\right)     $}: \\
This is a mutation from type $1$ to type $0$, which occurs 
at rate $(N-k)u^{}_N \nu^{}_0 $. 
We then have $k + 1$ fit individuals in the  population
altogether, but the new sample contains only $k$ of them. 
Arguing as in (c) and this time using \eqref{fitfix}, we get
\begin{align*}
&\mathbb{P} \left(  \text{absorption in} \ \mathcal{A}^{}_N \mid \left(\boldsymbol{M}^{}_0, Z^N_0 \right)
= \left((k , 0), k+1\right)\right)   \\
&=h^N_{k+1}  - \frac{h^N_{k+1}}{k+1} .
	\end{align*}
\end{description}
Note that, in steps (c) and (d) (and already in \eqref{fitfix} and \eqref{unfitfix}),
we have used the permutation invariance of the fit (respectively unfit)
lines to express the absorption probabilities as a function of $k$ (the number of
fit individuals in the population) alone. This way, we need not cope with the full
state space of $(\bM_t,Z^N_t)$.
Taking together the first-step principle with the results of (a)--(d),
we obtain the linear system of equations for  $h^N$ (with the rates $\lambda_k^N$ and $\mu_k^N$ as in \eqref{eq:lambda_mu}):
\begin{equation}
\left( \lambda^N_k + \mu^N_k \right) h_k^N = \lambda_k^N h^N_{k+1} + \mu_k^N h_{k-1}^N + k u_N^{} \nu_1^{} \frac{ 1 - h_{k-1}^N }{N-(k-1)} - (N-k) u_N^{} \nu_0^{} \frac{h_{k+1}^N}{k+1},
\label{h^Ngleichung}
\end{equation}
$0 < k < N$, which is complemented by the boundary conditions
$h_0^N=0$, $h_N^N=1$.
Rearranging results in 
\begin{equation}
\begin{split}
&\frac{1}{2} \frac{1}{N} \left( \lambda^N_k + \mu^N_k \right) N^2 \left(h^N_{k+1} - 2 h^N_{k} + h^N_{k-1} \right) \\
&+ \frac{1}{2}  \left( \lambda^N_k - \mu^N_k \right) \left( N \left( h^N_{k+1} - h^N_{k} \right) - N \left( h^N_{k-1} - h^N_{k}  \right)  \right) \\
&+ \frac{k}{N} \frac{N}{N-(k-1)} N u_N^{} \nu_1^{} \left( 1- h^N_{k-1} \right) - \frac{N-k}{N} \frac{N}{k+1} N u_N^{} \nu_0^{} h^N_{k+1}=0.
\end{split}
\label{h^Ngleichung2}
\end{equation}
Let us now consider a sequence $\left( k_N^{} \right)_{N \in
  \mathbb{N}}$ with $0 < k_N^{} < N$ and $\lim_{N \to \infty}
\frac{k_N^{}}{N} = x$. The probabilities
$h^N_{k_N}$ converge to $h(x)$ as $N \to \infty$ (for the stationary case a proof is given in the
Appendix). Equation (\ref{h^Ngleichung2}), with $k$ replaced by $k_N^{}$,
together with (\ref{drift coefficient}) and (\ref{diffusion
  coefficient}) leads to  
Taylor's boundary value problem (\ref{h DGL}).

\textbf{Equations for $\psi$.} As before, we consider $ \left(\boldsymbol{M}^{}_t, Z^N_t
\right)_{t \geqslant 0}$ with start in $\left( (k, 0), k
\right)$, and now introduce the new function $\psi^N_{k} := h^N_{k} -
\frac{k}{N}$. $\psi^N_{}$ is the part of the absorption probability in
$\mathcal{A}^{}_N$ that goes back to selective
reproductions (in comparison  to the neutral case). 
We therefore speak of $\psi^N_{}$ (as well as of $\psi$) as the
`extra' absorption probability.  

Substituting
$h^N_{k} =\psi^N_{k} + \frac{k}{N}$ in \eqref{h^Ngleichung} yields the
following difference equation for $\psi^N_{}$:
\begin{equation}
\begin{split}
\left( \lambda^N_k + \mu^N_k \right) \psi_k^N =
&\lambda_k^N \psi^N_{k+1} + \mu_k^N \psi_{k-1}^N + \frac{k(N-k)}{N²}s_N^{} \\
&- k u_N^{} \nu_1^{} \frac{ \psi_{k-1}^N }{N-(k-1)} - (N-k) u_N^{} \nu_0^{} \frac{\psi_{k+1}^N}{k+1}
\label{psi^N gleichung}
\end{split}
\end{equation}
$(0 < k < N)$, together with the boundary conditions $\psi_0^N = \psi_N^N = 0$.
It  has a nice
interpretation, which is completely analogous to that of $h^N$
except in case (a2): If one
of the fit sample individuals reproduces via a selective
reproduction event, the extra absorption probability is
$\psi^N_{k+1} + \frac{1}{N}$ (rather than $h^N_{k+1}$). Here, 
$\frac{1}{N}$ is the neutral fixation probability of the individual just
created via the selective event; $\psi^N_{k+1}$ is the
extra absorption probability of all $k + 1$ type-$0$
individuals present after the event. The neutral contribution gives rise to the $k(N-k)s_N/N^2$ term on the
right-hand side of \eqref{psi^N gleichung}.
Performing $N \to \infty$ in the same way as for $h$, we obtain
Taylor's boundary value problem (\ref{psi DGL}) and now have
an interpretation in terms of the graphical representation to go with it.

\subsection{Solution of the difference equation} 
\label{Solution of the difference equation}
In this Section, we derive an explicit expression for the fixation
probabilities $h^N_{k}$, that is, a solution of the difference
equation (\ref{h^Ngleichung}), or equivalently, (\ref{psi^N
  gleichung}). Although the calculations only involve standard
techniques, we perform them here explicitly since this yields
additional insight. Since there is no danger of confusion, we 
omit the subscript (or superscript) $N$ for economy of notation.

The following Lemma specifies the extra absorption
probabilities $\psi_{k}$ in terms of a recursion.
\begin{lemma}
Let $k \geqslant 1$. Then
\begin{equation}
 \psi_{N-k}= \frac{k(N-k)}{\mu_{N-k}^{}} \left( \frac{\mu^{}_{N-1}}{N-1} \psi_{N-1} + \frac{\lambda_{N-k+1}}{(k-1)(N-k+1)} \psi_{N-k+1} - \frac{ s (k-1)}{N²}\right).
 \label{psi rekursion}
\end{equation}

 \label{lemma rekursion}
\end{lemma}

\begin{remark}\label{remark raender}
The quantity $\lambda_k/(k(N-k)) = (1 + s)/N + u
  \nu_0^{}/k$ is well defined for all $1 \leqslant k \leqslant N$, and
 $k(N-k)/\mu_k^{} = (N-k)/(\frac{N-k}{N} + u
  \nu_1^{})$ is well defined even for $k =0$.
\end{remark}

\begin{proof}[Proof of Lemma $\ref{lemma rekursion}$]
Let $1 < i < N-1$. Set $k=i$ in (\ref{psi^N gleichung}) and
divide by $i(N-i)$ to obtain
\begin{align}
 \left( \frac{\lambda_i}{i(N-i)} + \frac{\mu_i^{}}{i (N-i)}  \right) \psi_i 
 &= \left( \frac{1+s}{N} + \frac{u \nu_0^{}}{i+1} \right) \psi_{i+1} + \left( \frac{1}{N} + \frac{u \nu_1^{}}{N-(i-1)} \right) \psi_{i-1} + \frac{s}{N²} \nonumber \\
 &= \frac{\lambda_{i+1}}{(i+1)(N-i-1)} \psi_{i+1} + \frac{\mu_{i-1}^{}}{(i-1)(N-i+1)} \psi_{i-1} + \frac{s}{N²}.
 \label{magic equation}
\end{align}
Together with
\begin{align}
 \left( \frac{\lambda_1}{N-1} + \frac{\mu_1^{}}{N-1}  \right) \psi_1 
 &= \frac{\lambda_{2}}{2(N-2)} \psi_{2} + \frac{s}{N²}, \label{magic equation1} \\
 \left( \frac{\lambda_{N-1}}{N-1} + \frac{\mu_{N-1}^{}}{N-1}  \right) \psi_{N-1} 
 &= \frac{\mu_{N-2}^{}}{2(N-2)} \psi_{N-2} + \frac{s}{N²}, \label{magic equationN}
\end{align}
and the boundary conditions $\psi_0 = \psi_N =0$, we obtain a new
linear system of equations for the vector $\psi=(\psi_k)_{0 \leqslant k \leqslant N}. $ 
Summation over the last $k$
equations yields
\begin{align*}
 \sum_{i=N-k+1}^{N-1} \left( \frac{\lambda_i}{i(N-i)} + \frac{\mu_i^{}}{i (N-i)}  \right) \psi_i 
 =& \sum_{i=N-k+1}^{N-2} \frac{\lambda_{i+1}}{(i+1)(N-i-1)} \psi_{i+1} \\
 &+ \sum_{i=N-k+1}^{N-1} \frac{\mu_{i-1}^{}}{(i-1)(N-i+1)} \psi_{i-1} + \frac{s(k-1)}{N²},
\end{align*}
which proves the assertion.

\end{proof}    \noindent
Lemma \ref{lemma rekursion} allows for an explicit solution
for $\psi$.
\begin{theorem}
For $1 \leqslant \ell,n \leqslant N-1$, let
\begin{equation}
\chi^n_{\ell}:=\prod_{i=\ell}^{n} \frac{\lambda_i}{\mu_i^{}} \text{ \ and \ }
K:=   \sum_{n=0}^{N-1}  \chi_1^n .
\label{K}
\end{equation}
The solution of recursion $(\ref{psi rekursion})$ is then given by
 \begin{equation}
  \psi_{N-k} = \frac{k(N-k)}{\mu_{N-k}^{}}  \sum_{n=N-k}^{N-1} \chi_{N-k+1}^n 
  \left( \frac{\mu_{N-1}^{}}{N-1} \psi_{N-1}  - \frac{s (N-1-n)}{N²} \right)
\label{psi_N-k}  
\end{equation}
with 
 \begin{equation}
 \psi_{N-1}  = \frac{1}{K} \frac{N-1}{\mu_{N-1}^{}} \frac{s}{N²}
 \sum_{n=0}^{N-2} (N-1-n)\chi_{1}^n . 
 \label{psi_N-1} 
\end{equation}
An alternative representation is given by
\begin{equation}
  \psi_{N-k} =\frac{1}{K} \frac{k(N-k)}{\mu_{N-k}^{}}  \frac{s}{N²}  
  \sum_{\ell=0}^{N-k-1} \sum_{n=N-k}^{N-1} (n-\ell) \chi_1^{\ell} \chi_{N-k+1}^n  .
 \label{psi_N-k alt}
\end{equation}
\label{Rekursionsloesung psi}
\end{theorem}

\begin{proof}
We first prove \eqref{psi_N-k} by induction over $k$. 
For $k=1$, \eqref{psi_N-k} is easily checked to be true.
Inserting the induction hypothesis for some $k-1\geqslant 0$ into 
recursion \eqref{psi rekursion} yields
\begin{align*}\label{ind1}
 \psi_{N-k}= & \frac{k(N-k)}{\mu_{N-k}^{}}
 \Biggl[  \frac{\mu_{N-1}^{}}{N-1} \psi_{N-1} \\
& + \frac{\lambda_{N-k+1}}{\mu_{N-k+1}^{}}  \sum_{n=N-k+1}^{N-1} \chi_{N-k+2}^n
\left(  \frac{\mu_{N-1}^{}}{N-1} \psi_{N-1}  - \frac{s (N-1-n)}{N²} \right)  - \frac{s (k-1)}{N²} \Biggr],
\end{align*}
which immediately leads to \eqref{psi_N-k}. For $k=N$, \eqref{psi_N-k} 
gives \eqref{psi_N-1}, since $\psi_0 = 0$ and $k(N-k)/\mu_{N-k}^{}$ is well defined
by Remark \ref{remark raender}. 
We now check \eqref{psi_N-k alt} by inserting \eqref{psi_N-1}
into \eqref{psi_N-k} and then use the expression for $K$ as in \eqref{K}:
\begin{align*}
 \psi_{N-k}&= \frac{1}{K} \frac{k(N-k)}{\mu^{}_{N-k}} \frac{s}{N²}
 \sum_{n=N-k}^{N-1} \chi_{N-k+1}^n
 \Biggl[ \sum_{\ell = 0}^{N-1}(N-1-\ell)  \chi_1^{\ell}  
 - \sum_{\ell = 0}^{N-1}(N-1-n) \chi_1^{\ell} 
 \Biggr] \\
 &=\frac{1}{K} \frac{k(N-k)}{\mu_{N-k}^{}} \frac{s}{N²}
 \sum_{\ell = 0}^{N-1}
  \sum_{n=N-k}^{N-1}
  (n - \ell) \chi_1^{\ell} \chi_{N-k+1}^n.
\intertext{Then we split the first sum according to whether
$\ell \leqslant N-k-1$ or $\ell \geqslant N-k$, and use 
$\chi_1^{\ell} =\chi_1^{N-k} \chi_{N-k+1}^{\ell} $
in the latter case:}
\psi_{N-k}&= \frac{1}{K} \frac{k(N-k)}{\mu_{N-k}^{}} \frac{s}{N²}
\Biggl[\sum_{\ell = 0}^{N-k-1} \sum_{n=N-k}^{N-1}
(n - \ell) \chi_1^{\ell} \chi_{N-k+1}^n \\
& \hphantom{=} + \chi_1^{N-k}\sum_{\ell = N-k}^{N-1} \sum_{n=N-k}^{N-1}
(n - \ell) \chi_{N-k+1}^{\ell} \chi_{N-k+1}^n
\Biggr].
\end{align*}
The first sum is the right-hand side of \eqref{psi_N-k alt}
and the second sum disappears due to symmetry.

\end{proof} \noindent
Let us note that the fixation probabilities thus obtained have been well known
for the case with selection in the absence of mutation (see, e.g., \cite[Thm. 6.1]{Durrett}),
but to the best of our knowledge, have not yet appeared in the literature for
the case with mutation.

\subsection{The solution of the differential equation} 
As a little detour, let us revisit the boundary value problem
(\ref{h DGL}). To solve it, Taylor assumes 
that $h$ can be expanded in a power series in $\sigma$. This yields a 
recursive series of boundary value problems (for the various powers of
$\sigma$), which are solved by elementary methods and combined into a 
solution of $h$ (cf. \cite{Taylor}). 

However, the calculations are slightly long-winded. In what follows, we show that
the boundary value problem (\ref{h DGL}) (or equivalently (\ref{psi DGL})) may
be solved in a direct and elementary way, without the need for a series expansion.
Defining
\begin{equation*}
 c(x):= - \theta \nu^{}_1 \frac{x}{1 - x} - \theta \nu^{}_0 \frac{1-x}{x} 
\end{equation*}
and remembering the drift coefficient $a(x)$ (cf. (\ref{drift coefficient})) and
the diffusion coefficient $b(x)$ (cf. (\ref{diffusion coefficient})), differential
equation (\ref{psi DGL}) reads
\begin{equation*}
 \frac{1}{2} b(x) \psi''\left( x \right) + a(x) \psi'\left( x \right) + c(x) \psi(x) = - \sigma x (1-x)  
 \end{equation*}
 or, equivalently,
 \begin{equation}
  \psi'' \left( x \right) + 2 \frac{a(x)}{b(x)} \psi'\left(x\right) + 2 \frac{c(x)}{b(x)} \psi(x) = - \sigma.
\label{psi DGL 2}
\end{equation}
Since 
\begin{equation}
  \frac{c(x)}{b(x)} = \frac{d}{dx} \frac{a(x)}{b(x)}, 
 \label{alternating sum}
\end{equation}
(\ref{psi DGL 2}) is an exact differential equation (for the concept of exactness, see
\cite[Ch. 3.11]{Ford} or \cite[Ch. 2.6]{Birkhoff}). Solving it corresponds
to solving its primitive
\begin{equation}
 \psi'\left(x \right) + 2 \frac{a(x)}{b(x)} \psi(x) = - \sigma (x - \tilde{x}).
\label{DGL Stammfunktion}
\end{equation}
The constant $\tilde{x}$ plays the role of an integration constant and
will be determined by the initial conditions later. (Obviously, \eqref{psi DGL 2}
is recovered by differentiating \eqref{DGL Stammfunktion} and observing
\eqref{alternating sum}.) As usual, we consider the homogeneous equation
\begin{equation*}
 \varphi'\left(x \right) + 2 \frac{a(x)}{b(x)} \varphi(x) = \varphi'\left(x \right) + \left( \sigma - \frac{\theta \nu_1^{}}{1-x} + \frac{\theta \nu_0^{}}{x} \right)\varphi(x) =0
\end{equation*}
first. According to \cite[Ch. 7.4]{Durrett} and \cite[Ch. 4.3]{Ewens}, its solution $\varphi_1^{}$ is
given by
\begin{equation*}
 \varphi_1^{}(x) = \exp \left( \int^x -2 \frac{a(z)}{b(z)} dz \right) 
 = \gamma \left(1-x \right)^{- \theta \nu_1^{}} x^{- \theta \nu_0^{}} \exp(- \sigma x)
 = \frac{2C \gamma}{b(x) \pi_X^{}(x)}  .
\end{equation*}
(Note the link to the stationary distribution provided by the last
expression (cf. \cite[Thm. 7.8]{Durrett} and \cite[Ch. 4.5]{Ewens}).) 
Of course, the same expression 
is obtained via separation of variables. Again we will deal with the
constant $\gamma$ later. \\
Variation of parameters yields the solution $\varphi_2^{}$ of the
inhomogeneous equation (\ref{DGL Stammfunktion}):
\begin{equation}
 \varphi_2^{} (x)
 = \varphi_1^{}(x) \int_{\beta}^x \frac{- \sigma (p - \tilde{x})}{\varphi_1^{}(p)}dp 
 = \sigma \varphi_1^{} (x) \int_{\beta}^x \frac{\tilde{x} - p}{\varphi_1^{}(p)} dp.
\label{solution phi}
\end{equation}
Finally, it remains to specify the constants of integration $\tilde{x}$,
$\gamma$ and the constant $\beta$ to comply with $\varphi_2^{}(0)= \varphi_2^{}(1)=0$.
We observe that the factor $\gamma$ cancels in (\ref{solution phi}), thus 
its choice is arbitrary. $\varphi_1^{}(x)$ diverges for $x \to 0$ and 
$x \to 1$, so the choice of $\beta$ and $\tilde{x}$ has to guarantee
$B(0)=B(1)=0$, where $B(x)=\int_{\beta}^x \frac{\tilde{x} - p}{\varphi_1^{}(p)} dp$.
Hence, $\beta =0$ and
\begin{equation*}
 \tilde{x} \int_0^1 \frac{1}{\varphi_1^{}(p)}dp 
 = \int_0^1 \frac{p}{\varphi_1^{}(p)}dp \ \ \Leftrightarrow \ \ \tilde{x} 
 = \frac{\int_0^1 \frac{p}{\varphi_1^{}(p)}dp }{\int_0^1 \frac{1}{\varphi_1^{}(p)}dp}.
\end{equation*}
For the sake of completeness, l'H\^{o}pital's rule can be used to check that
$\varphi_2^{}(0)= \varphi_2^{}(1)=0$. The result indeed coincides with Taylor's (cf. (\ref{solution h})).\\
We close this Section with a brief consideration of the initial value $a_1$
of the recursions \eqref{a_n}. Since, by \eqref{a_1}, $a_1 = - \psi'(1)$,
it may be obtained by analysing the limit $x \to 1$ of \eqref{DGL Stammfunktion}.
In the quotient $a(x)\psi(x)/b(x)$, numerator and denominator disappear
as $x \to 1$. According to l'H\^{o}pital's rule, we get
\begin{align*}
 \lim_{x \to 1} \frac{a(x)\psi(x)}{b(x)} 
=\lim_{x \to 1} \frac{(- \theta \nu_0^{} - \theta \nu_1^{} + \sigma (1-2x)) \psi(x) + a(x)\psi'(x) }{2(1-2x)} 
   = \frac{1}{2} \theta \nu_1^{} \psi'(1),
\end{align*}
therefore, the limit $x \to 1$ of \eqref{DGL Stammfunktion} yields
\begin{equation*}
- \psi'(1)(1 + \theta \nu_1^{})= \sigma(1-\tilde{x}).
\end{equation*}
Thus, we obtain $a_1$ 
without the need to differentiate expression \eqref{solution h}.

\section{Derivation of Fearnhead's coefficients in the discrete setting}
\label{Derivation of Fearnhead's coefficients in the discrete setting}
Let us now turn to the ancestral type distribution and Fearnhead's 
coefficients that characterise it. To this end, we start from the 
linear system of equations for $\psi^N=(\psi^N_k)_{0 \leqslant k \leqslant N} $ in 
\eqref{magic equation}-\eqref{magic equationN}. Let 
\begin{equation}\widetilde{\psi}^N_k:=\frac{ \psi^N_k}{k(N-k)} , 
 \label{psi tilde}
 \end{equation}
for $1 \leqslant k \leqslant N-1$. 
In terms of these new variables, \eqref{magic equationN} reads
\begin{equation}
 -\mu^N_{N-1} \widetilde{\psi}^N_{N-1}+\mu^N_{N-2} \widetilde{\psi}^N_{N-2} - \lambda^N_{N-1} \widetilde{\psi}^N_{N-1}
 + \frac{s_N^{}}{N^2} =0 .
 \label{magic equationN 2}
\end{equation}
We now perform linear combinations of \eqref{magic equation}-\eqref{magic equationN} 
(again expressed in terms of the $\widetilde{\psi}_{N-k}^N$) to obtain
\begin{equation}
 \begin{split}
  &\sum_{k=1}^{n-1} (-1)^{n-k-1} \binom{n-2}{k-1} (\lambda^N_{N-k} + \mu^N_{N-k}) \widetilde{\psi}^N_{N-k} \\
 &= \sum_{k=2}^{n-1} (-1)^{n-k-1} \binom{n-2}{k-1} \lambda^N_{N-k+1}  \widetilde{\psi}^N_{N-k+1} 
+ \sum_{k=1}^{n-1} (-1)^{n-k-1} \binom{n-2}{k-1} \mu^N_{N-k-1}  \widetilde{\psi}^N_{N-k-1} \\
 & \hphantom{=}+ \frac{s_N^{}}{N^2} \sum_{k=1}^{n-1} (-1)^{n-k-1} \binom{n-2}{k-1},
\label{sum magic equation}
\end{split}
\end{equation}
for $ 3 \leqslant n \leqslant N-1$. Noting that the last sum disappears as a consequence of the binomial theorem,
rearranging turns \eqref{sum magic equation} into
\begin{align}
 \sum_{k=0}^{n-1} (-1)^{n-k-1} \binom{n-1}{k}  \mu^N_{N-k-1} \widetilde{\psi}^N_{N-k-1}
 + \sum_{k=1}^{n-1} (-1)^{n-k} \binom{n-1}{k}  \lambda^N_{N-k} \widetilde{\psi}^N_{N-k}
 =0.
 \label{magic equation 2}
\end{align}
On the basis of equations \eqref{magic equationN 2}
and \eqref{magic equation 2} for $(\widetilde{\psi}^N_k)_{1 \leqslant k \leqslant N-1} $
we will now establish a discrete version of Fearnhead's coefficients, and a corresponding discrete version of recursion 
\eqref{a_n} and initial value \eqref{a_1}. Motivated by the limiting expression \eqref{h solution summe},
we choose the ansatz 
\begin{equation}
 \psi^N_{N-k} = (N-k) \sum_{i=1}^{k} a_i^N \frac{k_{[i]}}{N_{[i+1]}},
 \label{ansatz psi}
\end{equation}
where we adopt the usual notation $x_{[j]}:=x (x-1) \dots (x-j+1)$ for 
$x \in \mathbb{R}$, $j \in \mathbb{N}$.
Again we omit the upper (and lower) population size index $N$ (except for the one of the $a_n^N$) in 
the following Theorem.
\begin{theorem}
The $a_n^N$, $1  \leqslant n \leqslant N-1 $, satisfy the following relations: $a_1^N = N \psi_{N-1}$,
\begin{equation}
(N-2) \left[\left(\frac{2}{N} + u \nu_1^{}\right) a_2^N - 
\left(\frac{2}{N} + \frac{N-1}{N}s + u\right) a_1^N + \frac{N-1}{N}s\right]=0,
\label{diskrete Rekursion 2}
\end{equation}
and, for $3 \leqslant n \leqslant N-1$:
\begin{equation}
(N-n) \left[\left(\frac{n}{N} + u \nu_1^{}\right) a_n^N - 
\left(\frac{n}{N} + \frac{N-(n-1)}{N}s + u\right) a_{n-1}^N + \frac{N-(n-1)}{N}s a_{n-2}^N\right]=0.
\label{diskrete Rekursion n}
\end{equation}
\label{Theorem diskrete Rekursion}
\end{theorem}
\begin{proof}
At first, we note that the initial value $a_1^N$ follows directly from 
\eqref{ansatz psi} for $k=1$. 
Then we remark that, by \eqref{psi tilde} and \eqref{ansatz psi},
\begin{equation}
 \widetilde{\psi}_{N-k} = \frac{1}{k} \sum_{i=1}^{k} a_i^N \frac{k_{[i]}}{N_{[i+1]}}
 \label{ansatz}
\end{equation}
for $1 \leqslant k \leqslant N-1$. To prove \eqref{diskrete Rekursion 2}, we insert this into
\eqref{magic equationN 2} and write the resulting equality as
\begin{align*}
 \mu^{}_{N-2}a_2^N - (\mu^{}_{N-1}- \mu^{}_{N-2} + \lambda_{N-1})(N-2)a_1^N + \frac{(N-1)(N-2)}{N} s =0,
\end{align*}
which is easily checked to coincide with \eqref{diskrete Rekursion 2}. 

To prove \eqref{diskrete Rekursion n},
we express \eqref{magic equation 2} in terms of the $a_n^N$ via \eqref{ansatz}.
The first sum of \eqref{magic equation 2} becomes
\begin{align*}
  \sum_{k=0}^{n-1} (-1)^{n-k-1} \binom{n-1}{k}  \mu_{N-k-1}^{} \widetilde{\psi}_{N-k-1}
 & = \sum_{k=0}^{n-1} (-1)^{n-k-1} \binom{n-1}{k}  \mu_{N-k-1}^{} \sum_{i=1}^{k+1} a_i^N \frac{k_{[i-1]}}{N_{[i+1]}}  \\
 & = \sum_{i=1}^{n} a_i^N \sum_{k=i}^{n} (-1)^{n-k} \binom{n-1}{k-1}   \frac{(k-1)_{[i-1]}}{N_{[i+1]}} \mu_{N-k}^{} .
\end{align*}
Analogously, the second sum of \eqref{magic equation 2} turns into
\begin{align*}
 \sum_{k=1}^{n-1} (-1)^{n-k} \binom{n-1}{k}  \lambda_{N-k} \widetilde{\psi}_{N-k}
 = \sum_{i=1}^{n-1} a_i^N \sum_{k=i}^{n-1} (-1)^{n-k} \binom{n-1}{k}   \frac{(k-1)_{[i-1]}}{N_{[i+1]}} \lambda_{N-k} .
\end{align*}
Multiplying with $N!$, \eqref{magic equation 2} is thus reformulated as 
\begin{align}
  \sum_{i=1}^{n} a_i^N (N-i-1)_{[n-i]} (A_{\mu , i}^{n} + A_{\lambda , i}^{n})=0,
  \label{diskrete Rekursion}
\end{align}
where
\begin{align}
 A_{\mu , i}^{n} &:= \sum_{k=i}^{n} (-1)^{n-k} \binom{n-1}{k-1}   (k-1)_{[i-1]}  \label{A_mu} \mu_{N-k}^{} , \displaybreak[0] \\
 A_{\lambda , i}^{n} & := \sum_{k=i}^{n-1} (-1)^{n-k} \binom{n-1}{k}  (k-1)_{[i-1]}  \lambda_{N-k} . \label{A_lambda}
\end{align}
It remains to evaluate the $A_{\mu , i}^{n}$ and the $  A_{\lambda , i}^{n}$ for $1 \leqslant i \leqslant n$. 
First, we note that
$$
\binom{n-1}{k-1} (k-1)_{[i-1]} = \frac{(n-1)!}{(n-i)!} \binom{n-i}{k-i}
$$
for $i \leqslant k \leqslant n$ and apply this to the right-hand side of 
\eqref{A_mu}. This results in
\begin{align*}
 A_{\mu , i}^{n} &= \frac{(n-1)!}{(n-i)!} \sum_{k=i}^{n} (-1)^{n-k} \binom{n-i}{k-i}  \mu_{N-k}^{} 
 = \frac{(n-1)!}{(n-i)!} \sum_{k=0}^{n-i} (-1)^{k} \binom{n-i}{k}  \mu_{N-n+k}^{},
\end{align*}
where the sum corresponds to the $(n-i)$th difference quotient of the mapping
$$
\mu : \{ 0, \dots , N \}\rightarrow \mathbb{R}_{\geqslant 0}, \quad k \mapsto \mu_k^{} = - \frac{k^2}{N} + k(1+u \nu_1^{})
$$
taken at $N-n$. Since $\mu$ is a quadratic function, we conclude that $A_{\mu, i}^n=0$
for all $1 \leqslant i \leqslant n-3$. In particular, in the second difference quotient
(i.e., $i=n-2$) the linear terms cancel each other and 
$A^n_{\mu, n-2}$ simplifies to
\begin{align*}
A_{\mu, n-2}^n &= 
\frac{(n-1)!}{2} \big[ \mu_{N-n}^{} - 2 \mu_{N-n+1}^{} + \mu_{N-n+2}^{} \big] \displaybreak[0] \\
&= \frac{(n-1)!}{2} \big[- (N-n)^2 + 2 (N-n+1)^2 - (N-n+2)^2 \big]
=- \frac{(n-1)!}{N}  .
\end{align*}
For the remaining quantities $A_{\mu , n-1}^{n}$ and $A_{\mu , n}^{n}$, we have
$$
A_{\mu , n-1}^{n} = (n-1)! (\mu_{N-n}^{} - \mu_{N-n+1}^{})
= (n-1)! \left( \frac{1}{N} (N-2n+1) - u \nu_1^{} \right)
$$
and
$$
A_{\mu , n}^{n} = (n-1)! \mu_{N-n}^{} = (n-1)! (N-n) \left(\frac{n}{N} + u \nu_1^{}\right).
$$
We now calculate the $A^n_{\lambda, i}$. Since
$$
\binom{n-1}{k} (k-1)_{[i-1]} = \frac{1}{k} \frac{(n-1)!}{(n-1-i)!} \binom{n-1-i}{k-i}
$$
for $i \leqslant k \leqslant n-1$, we obtain that
\begin{align*}
A^n_{\lambda, i} 
&= \frac{(n-1)!}{(n-1-i)!} \sum_{k=i}^{n-1} (-1)^{n-k} \binom{n-1-i}{k-i} \frac{ \lambda_{N-k} }{k} \displaybreak[0] \\
&= \frac{(n-1)!}{(n-1-i)!} \sum_{k=0}^{n-1-i} (-1)^{k+1} \binom{n-1-i}{k} \frac{ \lambda_{N-(n-1-k)} }{n-1-k} ,
\end{align*}
where the sum now coincides with the $(n-1-i)$th difference quotient of the
affine function
$$
\lambda : \{ 0, \dots, N-1 \} \rightarrow \mathbb{R}_{\geqslant 0}, \quad k \mapsto \frac{\lambda_{k}}{N-k} = \frac{k}{N} (1+s) + u \nu_0^{}
$$
taken at $N-(n-1)$. Consequently, $A^n_{\lambda, i}=0$ for all 
$1 \leqslant i \leqslant n-3$, and in $A^n_{\lambda, n-2}$ (more precisely in the first difference quotient of
$\lambda$ at $N-(n-1)$) the constant terms cancel each other. Thus,
\begin{align*}
A^n_{\lambda, n-2}
&= (n-1)! \left[ - \frac{\lambda_{N-(n-1)}}{n-1} + \frac{\lambda_{N-(n-2)}}{n-2} \right] \displaybreak[0] \\
&= (n-1)! \frac{1+s}{N} [N- (n-2) -(N- (n-1)) ]  
= (n-1)! \frac{1+s}{N}
\end{align*}
and so
$$
A_{\lambda, n-1}^n 
= - (n-1)! \frac{\lambda_{N-(n-1)}}{n-1}
= -(n-1)! \left[ \frac{N-(n-1)}{N} (1+s) + u \nu_0^{}\right].
$$
Combining \eqref{diskrete Rekursion} with the results for $A^n_{\mu, i}$
and $A^n_{\lambda, i}$ yields the assertion \eqref{diskrete Rekursion n}.

\end{proof}
It will not come as a surprise now that
the discrete recursions of the $a_n^N$ obtained in Thm. \ref{Theorem diskrete Rekursion}
lead to Fearnhead's coefficients $a_n$ in the limit $N \to \infty$. According to Thm. 
\ref{Konvergenztheorem} in the Appendix, $ \psi^N_{k_N}$ converges to $\psi(x)$ 
for 
any given sequence $\left( k_N^{} \right)_{N \in \mathbb{N}}$ with $0 < k_N^{} < N$
and $\lim_{N \to \infty} \frac{k_N}{N} = x$.
Comparing \eqref{ansatz psi} with \eqref{h solution summe},
we obtain
$$
\lim_{N \to \infty} a_n^N = a_n
$$
for all $n \geqslant 1$. The recursions \eqref{a_n} of Fearnhead's coefficients 
then follow directly from the recursions in Thm. \ref{Theorem diskrete Rekursion}
in the limit $N \to \infty$.

\section{Discussion}
More than fifteen years after the discovery of the ancestral selection graph
by Neuhauser and Krone \cite{Krone,Neuhauser},
ancestral processes with selection constitute an active area of
research, see, e.g., the recent contributions
\cite{EtheridgeGriffiths,EtheridgeGriffithsTaylor,Mano,Pokalyuk,Vogl}.  
Still, the ASG  remains a challenge:
Despite the elegance and intuitive appeal of the concept, it is difficult to
handle when it comes to concrete applications. Indeed,
only very few properties of
genealogical processes in mutation-selection balance could be
described explicitly until today (see the conditional ASG
\cite{Wakeley,WakeleySargsyan} for an example).
Even the special case of a single ancestral line (emerging from a
sample of size one) is not yet fully understood. The
work by Fearnhead \cite{Fearnhead} and Taylor \cite{Taylor} established
important results about the CAP 
with the help of diffusion theory and analytical tools,
but the particle representation can only be partially recovered behind the
continuous limit.
In this article, we have therefore made a first step towards
complementing the picture by
attacking the problem from the discrete (finite-population) side.
Let us briefly summarise our
results.

The pivotal quantity considered here is the \emph{fixation probability
of the offspring of all fit individuals}, regardless of the types of the offspring.
Starting from the particle picture and using  elementary arguments
of first-step analysis, we obtained a difference equation for these
fixation probabilities. In the limit $N \to \infty$, the equation
turns into the (second-order ODE) boundary problem obtained via
diffusion theory by Taylor \cite{Taylor}, but now with an intuitive
interpretation attached to it.

We have given the solution of the difference equation in closed form;
the resulting fixation
probabilities provide a generalisation of the well-known finite-population
fixation probabilities in the case with selection only (note that they do not 
require the population to be stationary). As a little
detour, we also revisited the limiting continuous boundary value problem
and solved it via elementary methods, without the need of the series
expansion employed previously.

The fixation probabilities are intimately related with the
stationary type distribution on the ancestral line and can thus
be used for an alternative derivation of the recursions
that characterise Fearnhead's coefficients.
Fearnhead obtained these recursions by guessing and
direct (but technical) verification of the stationarity
condition; Taylor derived them in a
constructive way by inserting the ansatz \eqref{v} into the boundary value
problem \eqref{psi DGL} and performing a somewhat tedious differentiation
exercise. Here we have taken a third route 
that relies on the difference equation \eqref{magic equation} and stays
entirely within the discrete setting.

Altogether, the finite-population results 
contain more information than those obtained within
the diffusion limit; first, because they are not restricted to weak
selection, and second, because they are more directly related to the
underlying particle picture.
Both motivations also underlie, for example,
the recent work by Pokalyuk and Pfaffelhuber
\cite{Pokalyuk}, who re-analysed the process of
fixation under strong selection (in the absence of mutation)
with the help of an ASG. 

Clearly, the present article is only a first step towards a better
understanding
of the particle picture related to the common ancestor process.
It is known already that the coefficients $a_n$ may be interpreted
as the probabilities that there are $n$ virtual branches in the
pruned ASG at stationarity (see Section \ref{Fearnhead's approach});
but the genealogical content of the  recursions (15)
remains to be elucidated. It would also be desirable to
generalise the results to finite type spaces, in the spirit
of Etheridge and Griffiths \cite{EtheridgeGriffiths}.

\subsection*{Acknowledgement}
It is our pleasure to thank Anton Wakolbinger for enlightening discussions, 
and for Fig.~\ref{aline}. We are grateful to Jay Taylor for valuable comments on the manuscript,
and to Barbara Gentz for pointing out a gap in an argument at an earlier stage of the work.
This project received financial support by Deutsche Forschungsgemeinschaft
(DFG-SPP 1590), Grant no. BA2469/5-1.

\section*{Appendix}
In Section \ref{Derivation of differential equations} we have presented
an alternative derivation of the boundary value problem for the conditional
probability $h$. It remains to prove that $\lim_{N \to \infty} h^N_{k_N} = h(x)$,
with $x \in [0,1]$, $0 < k_N^{} < N$, $\lim_{N \to \infty} \frac{k_N}{N} = x$ 
and $h$ as given as in \eqref{solution h}. \\
Since $h_k^N = \frac{k}{N} + \psi^N_k$ and $h(x)=x + \psi(x)$, respectively, 
it suffices to show the corresponding convergence of the $\psi^N_k$.
For ease of exposition, we assume here that the process is stationary.
\begin{lemma}
Let $\tilde{x}$ be as in \eqref{x schlange}. Then
\begin{equation*}
\lim_{N \to \infty} N \psi^N_{N-1} = \frac{\sigma}{1 + \theta \nu_1^{}}(1-\tilde{x}).
\end{equation*}
\label{Konvergenz N psi N-1}
\end{lemma}

\begin{proof}
Since the stationary distribution $\pi^N_Z$ of $\left( Z^N_{t} \right)_{t \geqslant 0}$
(cf. \eqref{stationaere Verteilung Moran}) satisfies 
\begin{equation}
\prod_{i=1}^{n-1} \frac{\lambda_i^{N}}{\mu_i^{N}} 
= \frac{\pi^N_Z (n)}{C^{}_N} \frac{\mu^N_n}{\lambda^N_0}, 
\label{statVtlg}
\end{equation}
for $1 \leqslant n \leqslant N$, equation \eqref{psi_N-1} leads to
\begin{align*}
N \psi^N_{N-1} 
 &= \frac{Ns_N^{}}{1 + N u_N^{} \nu_1^{}}  \frac{\sum_{n=1}^{N} \pi^N_Z(n) \mu^N_{n} \frac{N-n}{N}}{\sum_{n=1}^{N}\pi^N_Z(n) \mu^N_{n}} \displaybreak[0] \\
 &=\frac{Ns_N^{}}{1 + N u_N^{} \nu_1^{}}  \frac{\sum_{n=1}^{N} \pi^N_Z(n) \frac{n(N-n)^2}{N^3} \left( 1 + \frac{N u_N^{} \nu_1^{}}{N-n} \right)}{\sum_{n=1}^{N}\pi^N_Z(n) \frac{n(N-n)}{N^2} \left( 1 + \frac{N u_N^{} \nu_1^{}}{N-n} \right) } ,
\end{align*}
where we have used \eqref{eq:lambda_mu} in the last step.
The stationary distribution of the rescaled process $\left( X^N_{t} \right)_{t \geqslant 0}$
is given by $\left( \pi^N_X \left( \frac{i}{N} \right)  \right)_{0 \leqslant i \leqslant N}$,
where $\pi^N_X \left( \frac{i}{N} \right) = \pi^N_Z (i)$. Besides, the sequence of
processes $\left( X^N_{t} \right)_{t \geqslant 0}$ converges to $(X_t)_{t \geqslant 0}$ in
distribution, hence
\begin{align*}
\lim_{N \to \infty} N \psi^N_{N-1} 
&= \lim_{N \to \infty} \frac{N s_N^{}}{1 + N u_N^{} \nu_1^{}} \frac{\mathbb{E}_{\pi^N_X}\left( X^N_{} \left( 1- X^N_{} \right)^2 \left( 1 + \frac{u_N^{} \nu_1^{}}{1-X^N_{}} \right) \right)}{\mathbb{E}_{\pi^N_X}\left( X^N_{} \left( 1- X^N_{} \right) \left( 1 + \frac{u_N^{} \nu_1^{}}{1-X^N_{}} \right) \right)} \\
&= \frac{\sigma}{1 + \theta \nu_1^{}} \frac{\mathbb{E}_{\pi_X^{}} \left( X (1-X)^2 \right)  }{\mathbb{E}_{\pi_X^{}} \left( X (1-X)\right)}
= \frac{\sigma}{1 + \theta \nu_1^{}}(1-\tilde{x}),
\end{align*}
as claimed.
\end{proof}

\begin{remark}
The proof gives an alternative way to obtain the initial value $a_1$
(cf. \eqref{a_1}) of recursion \eqref{a_n}.
\end{remark}

\begin{theorem}
For a given $x \in [0,1]$, let $\left( k_N^{} \right)_{N \in \mathbb{N}}$
be a sequence with $0 < k_N^{} < N$ and $\lim_{N \to \infty} \frac{k_N}{N} = x$. Then
\begin{equation*}
\lim_{N \to \infty} \psi_{k_N^{}}^N = \psi(x),
\end{equation*}
\label{Konvergenztheorem}
where $\psi$ is the solution of the boundary value problem \eqref{psi DGL}.
\end{theorem}

\begin{proof}
Using first Theorem \ref{Rekursionsloesung psi}, then \eqref{statVtlg},
and finally \eqref{eq:lambda_mu}, we obtain
\begin{align*}
\psi_k^N &= \frac{k(N-k)}{\mu_k^N} \sum_{n=1}^{N-k} \left( \prod_{i=k+1}^{N-n} \frac{\lambda^N_i}{\mu^N_i}  \right)  \left( \frac{\mu_{N-1}^N}{N-1} \psi^N_{N-1} - \frac{s_N^{} (n-1)}{N^2}  \right) \displaybreak[0] \\
&= \frac{k(N-k)}{\mu_k^N} \left( \mu^N_{k+1} \pi^N_Z (k+1)  \right)^{-1} \sum_{n=0}^{N-k-1} \mu^N_{N-n} \pi^N_Z (N-n) \left( \frac{\mu_{N-1}^N}{N-1} \psi^N_{N-1} - \frac{s_N^{} n}{N^2}  \right) \displaybreak[0] \\
&=\left( 1 + \mathcal{O}\left( \frac{1}{N} \right) \right) \left( \frac{k+1}{N} \frac{N-k-1}{N}   \pi^N_Z (k+1) \right)^{-1} \\
& \hphantom{=} \times \frac{1}{N} \sum_{n=0}^{N-k-1} \pi^N_Z (N-n) \frac{N-n}{N} \frac{n}{N} \left( 1 + \frac{Nu_N^{} \nu_1^{}}{n} \right) \left( (1 + N u_N^{} \nu_1^{}) N \psi^N_{N-1} - N s_N^{} \frac{n}{N} \right).
\end{align*}
In order to analyse the convergence of this expression, define
\begin{align*}
S_1^N (k) &:= \frac{k+1}{N} \frac{N-k-1}{N}  \pi^N_Z (k+1), \displaybreak[0] \\
S_2^N (k) &:= \frac{1}{N} \sum_{n=0}^{N-k-1} \pi^N_Z (N-n) \frac{N-n}{N} \frac{n}{N}  \left( (1 + N u_N^{} \nu_1^{}) N \psi^N_{N-1} - N s_N^{} \frac{n}{N} \right) \displaybreak[0] \\
& \hphantom{:} = \int_0^1 T^N_k (y) dy,  \displaybreak[0] \\
S_3^N (k) &:= \frac{1}{N} \sum_{n=0}^{N-k-1} \pi^N_Z (N-n) \frac{N-n}{N}   u_N^{} \nu_1^{} \left( (1 + N u_N^{} \nu_1^{}) N \psi^N_{N-1} - N s_N^{} \frac{n}{N} \right)\displaybreak[0] \\
& \hphantom{:} = \int_0^1 \tilde{T}_k^N (y) dy,
\end{align*}
with step functions $T^N_k : [0,1] \rightarrow \mathbb{R}$, $\tilde{T}^N_k : [0,1] \rightarrow \mathbb{R}$ given by
\begin{align*}
T^N_k (y)&:=\begin{cases}
 \mathbbm{1}_{ \lbrace n \leqslant N-k-1 \rbrace }\pi^N_Z (N-n) \frac{N-n}{N} \frac{n}{N}  \left( (1 + N u_N^{} \nu_1^{}) N \psi^N_{N-1} - N s_N^{} \frac{n}{N} \right),  \\
  \text{ \ if } \frac{n}{N} \leqslant y < \frac{n+1}{N}, n \in \lbrace 0, \dots , N-1 \rbrace , \\
 0,  \text{ if } y=1,
\end{cases} \displaybreak[0] \\
\tilde{T}^N_k (y) &:=  \begin{cases}
\mathbbm{1}_{ \lbrace n \leqslant N-k-1 \rbrace } \pi^N_Z (N-n) \frac{N-n}{N}   u_N^{} \nu_1^{} \left( (1 + N u_N^{} \nu_1^{}) N \psi^N_{N-1} - N s_N^{} \frac{n}{N} \right), \\
  \text{ \ if } \frac{n}{N} \leqslant y < \frac{n+1}{N}, n \in \lbrace 0, \dots , N-1 \rbrace , \\
 0,  \text{ if } y=1.
\end{cases}
\end{align*}
Consider now a sequence  $\left( k_N^{} \right)_{N \in \mathbb{N}}$ as in the
assumptions. Then $\lim_{N \to \infty} \pi^N_Z (k_N^{}) = \pi_X^{}(x)$
(cf. \cite[p. 319]{Durrett}), and due to Lemma \ref{Konvergenz N psi N-1}
\begin{align*}
\lim_{N \to \infty} S_1^N (k_N^{}) &= x (1-x) \pi^{}_X (x), \displaybreak[0] \\
\lim_{N \to \infty} T_{k_N^{}}^N (k_N^{}) 
&= \mathbbm{1}_{\lbrace y \leqslant 1-x \rbrace} \pi^{}_X (1-y)(1-y)y (\sigma (1- \tilde{x})- \sigma y), \displaybreak[0] \\
\lim_{N \to \infty} \tilde{T}_{k_N^{}}^N (k_N^{}) &= 0.
\end{align*}
Since $T^N_k$ and $\tilde{T}^N_k$ are bounded, we have
\begin{align*}
\lim_{N \to \infty} S_2^N (k_N^{}) &= \int_0^{1-x} \pi^{}_X (1-y)(1-y)y (\sigma (1- \tilde{x})- \sigma y) dy, \\
\lim_{N \to \infty} S_3^N (k_N^{}) &=0,
\end{align*}
thus
\begin{equation*}
\lim_{N \to \infty} \psi^N_{k_N^{}} = \left( x (1-x) \pi^{}_X (x) \right)^{-1} \int_0^{1-x} \pi^{}_X (1-y)(1-y)y (\sigma (1- \tilde{x})- \sigma y) dy.
\end{equation*}
Substituting on the right-hand side yields
\begin{align*}
\lim_{N \to \infty} \psi^N_{k_N^{}} 
&= \left( x (1-x) \pi^{}_X (x) \right)^{-1} \sigma \int_x^1 \pi_X^{}(y)y(1-y) (y - \tilde{x}) dy  \displaybreak[0] \\
&=\left( x (1-x) \pi^{}_X (x) \right)^{-1} \sigma \left[ \int_0^1 \pi_X^{} (y) y (1-y) (y - \tilde{x})dy + \int_0^x \pi_X^{} (y) y (1-y) ( \tilde{x}-y)dy \right] \displaybreak[0] \\
&=\left( x (1-x) \pi^{}_X (x) \right)^{-1} \sigma  \int_0^x \pi_X^{} (y) y (1-y) ( \tilde{x}-y)dy = \psi(x) ,
\end{align*}
where the second-last equality goes back to the definition of $\tilde{x}$
in \eqref{x schlange}, and the last is due to \eqref{solution h}, \eqref{def psi},
and \eqref{wright}.
\end{proof}

\end{document}